\documentclass[11pt]{article}
\usepackage{style}

\usepackage[margin=1in]{geometry}
\usepackage{tikz}
\usepackage{verbatim}
\usepackage{cancel}
\usepackage{titling}
\usepackage{multirow}
\usepackage{cancel}
\usepackage{quantikz}
\usepackage[colorinlistoftodos,prependcaption,textsize=tiny]{todonotes}

\providecommand{\compl}{\mathsf{compl}}
\providecommand{\ucc}{\mathsf{ucc}}
\providecommand{\cc}{\mathsf{cc}}

\providecommand{\ent}{\mathsf{Ent}}

\newcommand{\ignore}[1]{}
\definecolor{cb-salmon-pink}{RGB}{255, 182, 119}
\definecolor{ref-color}{RGB}{200, 0, 200}

\newcommandx{\unsure}[2][1=]{\todo[linecolor=red,backgroundcolor=red!25,bordercolor=red,#1]{#2}}
\newcommandx{\change}[2][1=]{\todo[linecolor=blue,backgroundcolor=blue!25,bordercolor=blue,#1]{#2}}
\newcommandx{\info}[2][1=]{\todo[linecolor=OliveGreen,backgroundcolor=OliveGreen!25,bordercolor=OliveGreen,#1]{#2}}
\newcommandx{\improvement}[2][1=]{\todo[linecolor=Plum,backgroundcolor=Plum!25,bordercolor=Plum,#1]{#2}}
\newcommandx{\thiswillnotshow}[2][1=]{\todo[disable,#1]{#2}}

\title{More Efficient $k$-wise Independent Permutations from Random Reversible Circuits via log-Sobolev Inequalities}

\author{
Lucas Gretta\thanks{UC Berkeley. Email: \url{lucas_gretta@berkeley.edu}.}
\and
 William He\thanks{Carnegie Mellon University. Email: \url{wrhe@cs.cmu.edu}. Supported in part by ARO grant W911NF2110001.}
 \and 
 Angelos Pelecanos\thanks{UC Berkeley. Email: \url{apelecan@berkeley.edu}.} }

\begin{document}
\allowdisplaybreaks
\maketitle
\begin{abstract}
    We prove that the permutation computed by a reversible circuit with $\wt{O}(nk\cdot \log(1/\epsilon))$ random $3$-bit gates is $\epsilon$-approximately $k$-wise independent.    
    Our bound improves on currently known bounds in the regime when the approximation error $\epsilon$ is not too small. 
    We obtain our results by analyzing the log-Sobolev constants of appropriate Markov chains rather than their spectral gaps.
\end{abstract}
\tableofcontents

\newpage
\section{Introduction}

We consider the extent to which small random reversible circuits compute \emph{almost $k$-wise independent} permutations. The (almost) $k$-wise independence of permutations was first considered by Gowers~\cite{gowers1996almost} as a proxy for pseudorandomness properties of practical cryptosystems, such as block ciphers. 
\begin{definition}[Approximate $k$-wise independent permutations]
    A distribution $\mathcal{P}$ on the symmetric group $S_{[N]}$ is said to be \emph{$\varepsilon$-approximate $k$-wise independent} if for all distinct $x_1, \dots, x_k \in [N]$, the distribution of $(\boldsymbol{g}(x_1), \dots, \boldsymbol{g}(x_k))$ for $\boldsymbol{g} \sim \mathcal{P}$ has total variation distance at most $\varepsilon$ from the uniform distribution on distinct $k$-tuples over~$[N]$. 
\end{definition}

A commonly studied construction of approximate $k$-wise independent permutations is a reversible circuit on $n$ wires in which each gate computes a randomly chosen width-2 (see \Cref{def:width 2}) permutation on a random subset of $3$ wires. From here on, when referring to a random reversible circuit, we mean a random circuit whose gates are drawn randomly from a set of $3$-bit gates. Gowers~\cite{gowers1996almost} introduced this construction and proved that a random reversible circuit with $\mathrm{poly}(n,k,\log(1/\varepsilon))$ gates computes an $\varepsilon$-approximate $k$-wise independent permutation of the cube $\{0,1\}^n$ using the canonical paths technique from Markov chain mixing~\cite{jerrum2003counting}. Since then, follow-up works by Hoory et al. and Brodsky and Hoory~\cite{HMMR05, BH08} improved on the analysis of Gowers and proved that if $k\leq 2^{n/50}$, then random reversible circuits with $O(n^2 k^2 \log(1/\varepsilon))$ gates compute an $\varepsilon$-approximate $k$-wise independent permutation using the comparison method~\cite{DSC93, DSC93b}. Finally, using quantum-inspired techniques for proving spectral gaps, He and O'Donnell~\cite{he2024pseudorandom} improved the number of gates needed to $ \wt{O}(nk)\cdot (nk+\log(1/\varepsilon))$. 

Random circuits have gained attention following the recent interest in random \emph{quantum} circuits. The natural quantum analog of a (approximate) $k$-wise independent permutation is that of a (approximate) unitary $k$-design.\footnote{A (approximate) unitary $k$-design is a distribution on the unitary group that (approximately) matches the Haar distribution up to $k^{th}$ moments.} Unitary designs are widely studied in quantum computation and quantum physics as basic pseudorandom objects and models for equilibration in quantum many-body systems~\cite{BCH21}. A line of work on unitary $k$-designs~\cite{brandao2016local,haferkamp2021improved} shows that for constant $\varepsilon$, a reversible circuit on $n$ wires with $\wt{O}(n^2\cdot \mathrm{poly}(k))$ random 3-qubit quantum gates chosen from some finite gate set (a random quantum circuit) gives a construction of an $\varepsilon$-approximate unitary $k$-design.

Recent works~\cite{MPSY24, CBBDHX24} obtain $k$-designs with size linear in $k$ from classical $k$-wise independent permutations whose size is also linear in $k$. Even though we demonstrate that a linear-in-$k$ number of random width-$2$ gates suffices to $\varepsilon$-approximate $k$-wise independence, we remark that our dependence on $\varepsilon$ is not sufficiently tight for their $k$-design construction. In particular, both works employ a theorem of Alon and Lovett~\cite{AL12} which requires an exponentially small $\varepsilon$ to translate from approximate to exact $k$-wise independent permutations. Plugging in such a small $\varepsilon$ in our theorem would increase our size bound by polynomial factors in $n$ and $k$.

Another line of work, motivated by the design of practical cryptosystems (such as block ciphers), studies the computational pseudorandomness properties of random reversible circuits. He and O'Donnell~\cite{he2024pseudorandom} consider the computational hardness of inverting the permutation computed by short reversible circuits with $3$-bit gates. Another line of work by Canetti et al.~\cite{CCMR24} proposed more advanced cryptographic primitives based on the cryptographic properties of random reversible circuits. In particular, using the assumption that random reversible circuits achieve computational pseudorandomness after a modest number of rounds (much less than the super-polynomial number of rounds required to reach statistical pseudorandomness), they suggest candidate obfuscation schemes along with possible ways to prove their computational security. Their approach is inspired by thermalizing processes of statistical mechanics. 

In this paper, we revisit the problem of random circuits with reversible $3$-bit gates and show that a random reversible circuit with $\wt{O}(nk\cdot\log(1/\varepsilon))$ gates gives an $\varepsilon$-approximate $k$-wise independent permutation. The following is our main theorem, which we prove in \Cref{sec:generic states}. 


\begin{restatable}[]{theorem}{main}\label{thm:main}
     For any $n$ and $k\leq 2^{n/50}$, a random reversible circuit with $\wt{O}(nk\cdot \log(1/\varepsilon))$ width-$2$ gates (a subset of $3$-bit gates) computes an $\varepsilon$-approximate $k$-wise independent permutation, where the $\wt{O}$ hides $\mathrm{polylog}(n,k)$ factors.
\end{restatable}

We note here that for applications of approximate $k$-wise independent permutation distributions~$\calP$ in derandomization, one is generally concerned with the number of truly random ``seed'' bits needed to generate a draw from~$\calP$. See, for example~\cite{mohanty2020explicit}. By using techniques such as derandomized squaring (see~\cite{kaplan2009derandomized}), one can often reduce the seed length to $O(nk)$ for any construction. This is true for the results in our paper, and we don't discuss the seed length any further, as we are generally focused on the circuit complexity of our permutations.

\subsection{Proof overview}
We use the comparison method in a similar way as~\cite{BH08}. In particular, we bound the log-Sobolev constant of the natural Markov chain associated with the computation of a random reversible circuit, by comparing it to the log-Sobolev constant of the \emph{$k$-clique $2^n$-coloring} Markov chain. By working with the log-Sobolev constant rather than the spectral gap of this random walk as~\cite{BH08, he2024pseudorandom} do, we obtain an improved mixing time since the log-Sobolev constant gives a mixing time bound that depends doubly logarithmically on the smallest probability of the stationary distribution. In contrast, the spectral gap gives bounds that depend logarithmically on this quantity.

While it is generally more difficult to bound the log-Sobolev constant of a Markov chain, recent work of Salez~\cite{Sal20} has used the martingale method of Lee and Yau~\cite{lee1998logarithmic} to obtain sharp estimates for the log-Sobolev constant of a natural random walk on the multislice. Using this method, we estimate the log-Sobolev constant of a variant of $k$-clique $2^n$-coloring chain, which we call the \emph{uniform $k$-clique $2^n$-coloring} chain. The log-Sobolev constant for the standard $k$-clique $2^n$-coloring chain is then obtained via a simple application of the comparison method.

In more detail, our starting point is the work of Salez which bounds the log-Sobolev of the \emph{multislice}. The multislice corresponds to the random walk over the set of colorings of $2^n$ items, where each step of the walk swaps the colors of any two items chosen uniformly at random. The colorings are comprised of $k+1$ colors, where the first $k$ colors appear once and the last color appears in the remaining $2^n - k$ items. The first observation is that this random walk captures the $k$-wise independence of a random walk with transpositions. Unfortunately, the log-Sobolev constant of this walk is too small: $\pbra{n\cdot 2^n}^{-1}$. In contrast, we would expect a random set of transpositions to mix to a $k$-wise independent permutation within a time that is dependent on $k$.

The reason that the log-Sobolev constant of the multislice chain is independent of $k$ is because it applies a random transposition from the entire set of $\binom{2^n}{2}$ transpositions. In the case when $k$ is much smaller than $2^n$, a random transposition will most likely exchange the colors of two of the $2^n-k$ items that have color $k+1$. Thus, with high probability, roughly $1 - \frac{k}{2^n}$, the multislice chain will not move to a new state. To avoid this artificial slowdown, we study the uniform $k$-clique $2^n$-coloring chain, which requires that every step applies one transposition with an element that doesn't have color $k+1$. Equivalently, one may think of the uniform $k$-clique $2^n$-coloring chain as a random walk on the multislice that takes $\frac{2^n}{k}$ steps per time step and thus would hope that the log-Sobolev constant scales down by a factor of $\frac{k}{2^n}$. Indeed, we employ the martingale method and prove that the log-Sobolev constant of the uniform $k$-clique $2^n$-coloring chain is $\Omega\pbra{\frac{1}{nk}}$ as expected.

One can compute the log-Sobolev constant of the uniform $k$-clique $2^n$-coloring chain by using Salez's result as a black box and viewing the multislice chain as a lazy version of the uniform $k$-clique $2^n$-coloring chain. We instead present an alternative proof by adapting the martingale method used by Salez.

The next step is to transfer our log-Sobolev bound from the uniform $k$-clique $2^n$-coloring chain to the $k$-clique $2^n$-coloring chain, which has slightly different transition probabilities than its uniform counterpart. We give a randomized paths construction with only a constant amount of congestion. The comparison method implies that the log-Sobolev constant of the $k$-clique $2^n$-coloring chain is also $\Omega\pbra{\frac{1}{nk}}$.

Finally, we obtain an estimate for the log-Sobolev constant of the random reversible circuits Markov chain by employing the comparison with the $k$-clique $2^n$-coloring chain
from~\cite{BH08}. More specifically, Brodsky and Hoory give a randomized paths construction with a comparison constant of $\Theta(n^2)$. This concludes our $\Omega\pbra{\frac{1}{n^3k}}$ bound for the log-Sobolev constant of the reversible circuits Markov chain.

To improve our bound on the mixing time of the reversible circuits Markov chain, we use another argument from~\cite{BH08}. The observation is that after a short random walk of $\wt{O}(n)$ steps, the state of the reversible circuits Markov chain is very likely to be in a \emph{generic} state. Thus it suffices to bound the mixing time of the Markov chain when restricted to generic states. We do this by bounding its log-Sobolev constant, using the log-Sobolev inequality of the clique coloring chain, which we proved earlier. This allows us to bring down the mixing time of the reversible circuits Markov chain to $O(nk\cdot\mathrm{polylog}(n,k))$.

\section{Preliminaries}

\paragraph{Notation.} In this paper we will use the symbols $\gtrsim, \lesssim$ to compare two quantities in the asymptotic sense, in particular, these symbols hide constant factors. For example, $f(n)\lesssim g(n) \iff f(n) \leq O(g(n))$. When $x=(x_1,\dots,x_k)$ is a tuple, we use the notation $\ell\in x$ whenever $\ell = x_i$ for some $i \in [k]$ and otherwise, we write $\ell\not\in x$.

\begin{definition}[Tuples with distinct elements]
    \label{def:distinct-tuples}
    Let $S$ be a set. We define the set of \emph{$k$-tuples with distinct elements from $S$} as follows:
    $$\Theta_{k, S} \coloneq \cbra{(x_1,\dots,x_k)\in S^k:x_i\text{'s distinct}}.$$
    We frequently write $\Theta_{k, N}$ in the place of $\Theta_{k, [N]}$.
\end{definition}

We recall the definition of width-$2$ simple permutations from~\cite{BH08}. 

\begin{definition}[Width-$2$ simple permutations]
\label{def:width 2}
    The set of \emph{width-$2$ simple permutations} is the following set of permutations on $\{0, 1\}^n$
    $$\Sigma \coloneq \left\{f_{i, j_1, j_2, h} :
    \begin{array}{c}
    i, j_1, j_2 \in [n], i \neq j_1, j_2 \\
    h~\text{Boolean function on}~\{0, 1\}^2
    \end{array}
    \right\}.$$
    The permutation $f_{i, j_1, j_2, h}$ maps $(x_1, \dots, x_n)$ to $(x_1, \dots, x_{i-1}, x_i\oplus h(x_{j_1}, x_{j_2}), x_{i+1}, \dots, x_n)$.
\end{definition}
In words, a width-$2$ permutation chooses $3$ random indices from $[n]$: $i$ and $j_1, j_2$. It further samples a random Boolean function on $2$ bits. Then it XORs the value of $h(x_{j_1}, x_{j_2})$ on the $i^{th}$ bit of the input.

\subsection{Log-Sobolev constant and mixing time}

We recall some background on Markov chains from \cite{SC97}. Let $P$ be the transition matrix of an ergodic Markov chain over finite state space $V$, and let $\pi$ denote its stationary distribution. We identify a Markov chain with its transition matrix, so we will often say that $P$ is both the transition matrix for a Markov chain and also the Markov chain itself. We let $p^t_x$ denote the probability distribution of $P$, starting at state $x$, at timestep $t$. 

\begin{definition}[Mixing time]
    The $\varepsilon$-mixing time of an ergodic Markov chain $P$ is defined as:
    \begin{align*}
        \tau_\varepsilon(P) \coloneq \min \cbra{t \geq 0 : \max_{x \in V} \norm{p_x^t - \pi}_{\text{TV}}}.
    \end{align*}
    When the subscript is dropped, we mean $\tau(P)=\tau_{1/4}(P)$.
\end{definition}

Throughout this paper, we deal only with \textit{reversible} Markov chains. 
\begin{definition}[Reversible Markov chain]
    We say that a Markov chain $P$ is \emph{reversible} if for all $x,y \in V$, 
    \begin{align*}
        \pi(x)P(x,y) = \pi(y)P(y,x).
    \end{align*}
\end{definition}

One powerful way of bounding the mixing time of Markov chains is by functional inequalities using the Dirichlet form.

\begin{definition}[Dirichlet form]
    \label{def:dirichlet}
    For function $f:V \to \mathbb{R}_{\geq0}$, the \emph{Dirichlet form} of $f$ with respect to $P$ is
    $$\mathcal{E}_P(f, f) \coloneq \frac{1}{2}\sum_{x, y\in \Omega}\left(f(x) - f(y)\right)^2 \pi(x)P(x, y).$$
\end{definition}

Intuitively, the Dirichlet form measures the ``local variation'' of $f$ with respect to the (weighted) graph underlying a Markov chain $P$.
\begin{definition}[Entropy]
    For a function $f:V \to \mathbb{R}_{\geq0}$, we define its \emph{entropy}
    $$\ent_\pi[f] \coloneq \sum_{x \in V} \pi(x)f(x)\log \frac{f(x)}{\mathbb{E}_\pi[f]},$$
    where $\mathbb{E}_\pi[f] = \sum_{x\in V} \pi(x)f(x)$.
\end{definition}

The ratio of these two quantities defines the log-Sobolev constant of the Markov chain.

\begin{definition}[Log-Sobolev constant of Markov chain]
    The \textit{log-Sobolev constant} of $P$ is defined by
    $$\alpha(P) \coloneq \inf_{\substack{f \geq 0 \\ f~\text{non-constant}}} \frac{\mathcal{E}_P(\sqrt f, \sqrt f)}{\ent_\pi[f]}.$$
\end{definition}

The log-Sobolev constant of a Markov chain bounds the mixing time of the chain according to the following theorem. Note the doubly-logarithmic dependence on $1/\pi_{\mathrm{min}}$, which is the conceptual advantage of using log-Sobolev inequalities over a spectral gap analysis, whenever $\varepsilon$ is not exponentially small.

\begin{theorem}[\cite{DSC96}, Theorem 3.7]
    \label{thm:log-sob-to-mixing}
    Let $P$ be the transition matrix of a reversible Markov chain whose stationary distribution is $\pi$, and $\pi_{\min}$ to be the smallest stationary probability. For $\varepsilon \leq \frac{1}{e}$, the $\varepsilon$-mixing time is bounded by
    \begin{align*}
        \tau_\varepsilon(P) \lesssim \frac{1}{\alpha}\pbra{\log\log  \frac{1}{\pi_{\min}} + \log \frac{1}{\varepsilon}}.
    \end{align*}
\end{theorem}

In fact, the log-Sobolev constant bounds the $\ell^\infty$ mixing time, which gives \textit{pointwise} distance bounds.

\begin{theorem}[\cite{DSC96}, Corollary 3.8]\label{thm:pointwiseconvergence}
     For reversible $P$, and for all $x,y \in V$
    $$\left|p_x^t(y) - \pi(y)\right| \leq \varepsilon \pi(y)$$
    when $t \gtrsim \frac{1}{\alpha}\pbra{\log \log \frac{1}{\pi_\text{min}} + \log \frac{1}{\varepsilon}}$.
\end{theorem}

\subsection{The comparison method}

We bound the log-Sobolev constant of a reversible circuits Markov chain by repeated application of the comparison method~\cite{DSC93,wilmer2009markov} which we introduce below. The comparison method is used to estimate the Dirichlet form of a \emph{target} Markov chain with transition matrix $P$ by relating it to the Dirichlet form of a \emph{reference} Markov chain with transition matrix $\wt{P}$, for which we have previously-known estimates. This relation between Dirichlet forms can be trivially extended to an inequality between log-Sobolev constants when $\wt{P}$ and $P$ are over the same state space $V$ and have the same stationary distribution $\pi$.

The comparison is achieved by ``simulating'' the transition probabilities of the $\wt{P}$ Markov chain using paths from $P$. Formally, for each $(x,y)\in V^2$ we assign a random path 
\begin{align*}
    \bm{\Delta}(x,y)=\pbra{(x,\bm{u}_1),(\bm{u}_1,\bm{u}_2),(\bm{u}_2,\bm{u}_3),\dots,(\bm{u}_{\bm{\ell}},y)},
\end{align*}
where the $\bm{u}_i$'s are random elements of $V$ that satisfy  $P(x, \bm{u}_1), P(\bm{u}_{\bm{\ell}}, y) >0$ and $P(\bm{u}_i, \bm{u}_{i+1}) > 0$. The quantity $\bm{\ell}$ is a random non-negative integer equal to the length of the path $|\bm{\Delta}(x, y)|$. The congestion of these paths (which is captured by the comparison constant $A(\bm{\Delta})$) provides a lower bound of $\mathcal{E}$ with respect to $\wt{\mathcal{E}}$ as shown formally in~\Cref{thm:comparison}.

Without loss of generality, we assume that the paths $\bm{\Delta}(x, y)$ are simple, since one can remove all loops without affecting the endpoints $x, y$ of a path and without increasing the congestion.

\begin{lemma}[\cite{wilmer2009markov}, Corollary 13.23]
    \label{thm:comparison}
    Let $\wt{P}$ and $P$ be transition matrices for two ergodic Markov chains on the same state space $V$. Assume that for each $(x,y)\in V^2$ there exists a random path 
    \begin{align*}
        \bm{\Delta}(x,y)=\pbra{(x,\bm{u}_1),(\bm{u}_1,\bm{u}_2),(\bm{u}_2,\bm{u}_3),\dots,(\bm{u}_{\bm\ell},y)}.
    \end{align*}
    Then we have for any $f:V\to\R$ that
    \begin{align*}
        \wt{\mathcal{E}}(f, f) \leq A(\bm{\Delta})\cdot \mathcal{E}(f, f)
    \end{align*}
    where the comparison constant of $\bm{\Delta}$ is defined to be
    \begin{align*}
        A(\bm{\Delta}) \coloneq \max_{\substack{(a,b)\in V^2 \\\wt{P}(a,b)>0}} \cbra{\frac{1}{\pi(x)P(a,b)}\sum_{(x,y)\in V^2}\Ex_{\bm{\Delta}}\sbra{\mathbf{1}_{(a, b) \in \bm{\Delta}(x,y)}\cdot |\bm{\Delta}(x,y)|}\cdot  \wt{\pi}(x)\cdot \wt{P}(x, y)}.
    \end{align*}
    Here $\pi$ and $\wt{\pi}$ are the (unique) stationary distributions for $P$ and $\wt{P}$, respectively, and $\mathbf{1}_{(a, b) \in Q}$ is the indicator variable which captures whether the edge $(a,b)$ appears in the sequence $Q$.
\end{lemma}


\section{The Markov chains}

We now set up the Markov chains we use in the proof of \Cref{thm:main}. Throughout this section (and the rest of the paper) fix positive integers $n$, $k$, and $N$ (which will typically be equal to $2^n$). Our Markov chains all have domains isomorphic to $\Theta_{k,U}$ for some set $U$:

\begin{definition}[Reversible circuit Markov chain]
\label{def:rev-circuit-chain}
    The chain $\cbra{\bm{X}^{\mathsf{rev}}_t}_{t\geq 0}$ on the state space of $k$ distinct $n$-bit strings
    is given by the following distribution on $\bm{X}_{t+1}^{\mathsf{rev}}|\bm{X}_{t}^{\mathsf{rev}}$. Given the current state $x=(x_1,\dots,x_k)$, to draw the next state $\bm{X}_{t+1}=(\bm{y}_1,\dots,\bm{y}_k)$, draw a uniformly random width-2 permutation $\bm\sigma \in \Sigma$
    and set 
    \begin{align*}
       (\bm{y}_1,\dots,\bm{y}_k)=(\bm\sigma x_1,\dots,\bm\sigma x_k).
    \end{align*}
    Let $P^{\mathsf{rev}}_{k, n}$ be the transition matrix of this Markov chain.
\end{definition}

This Markov chain exactly captures the evolution of $k$ inputs to a random reversible circuit whose gates are uniformly drawn from the set of width-$2$ permutations $\Sigma$. Thus the statement of~\Cref{thm:main} that a random reversible circuit with $s$ width-$2$ gates is an $\varepsilon$-approximate $k$-wise independent permutation is implied by the statement that $\tau_\varepsilon\pbra{P^{\mathsf{rev}}_{k, n}} \leq s$. We typically write $P^{\textsf{rev}}$ and omit the parameters $k$ and $n$ whenever they are clear from the context or not important.

Following~\cite{BH08}, we prove that this Markov chain mixes fast by comparing it to the \emph{$k$-clique $2^n$-coloring} Markov chain. In this paper we deal with two clique coloring chains, thus we will refer to this chain as the \emph{standard clique coloring}, or simply the \emph{clique coloring} chain. (Note that this chain is slightly different than the )

\begin{definition}[Standard $k$-clique $N$-coloring Markov chain]
    \label{def:clique-coloring}
   Let $N$ be the number of colors and $k$ be the number of clique vertices. The \emph{$k$-clique $N$-coloring} chain $\cbra{\bm{X}^{\cc}_t}_{t\geq 0}$ on the set of colorings $\Theta_{k, N}$ is given by the following distribution on $\bm{X}_{t+1}^{\cc}|\bm{X}_{t}^{\cc}$. To sample $\bm{X}_{t+1}^{\cc} =(\bm{y}_1,\dots,\bm{y}_k)$ given the current state $\bm{X}_{t}^{\cc}=x=(x_1,\dots,x_{k})$, uniformly sample $\bm{i}\in [k]$ and $\bm{\ell}\in \{\ell\in[N]:\ell \not\in x\} \cup \{x_{\bm{i}}\}$ and set
    \begin{align*}
        \bm{y}_{j}=& \begin{cases}
            \bm{\ell} & j = \bm{i} \\
            x_j & j \neq \bm{i}
        \end{cases}.
    \end{align*}
    Let $P^\cc_{k, N}$ be the transition matrix for this Markov chain.
\end{definition}

In other words, the clique coloring chain samples a uniformly random coloring of the $k$-clique with $N$ colors, by randomly choosing a vertex and randomly assigning it one of the $(N-k+1)$ available colors (including its current color).

We directly bound the log-Sobolev constant of a related Markov chain, which we call the \emph{uniform clique coloring} chain. 

\begin{definition}[Uniform $k$-clique $N$-coloring Markov chain]
    \label{def:uniform-clique-coloring}
   Let $N$ be the number of colors and $k$ be the number of clique vertices. The \emph{uniform $k$-clique $N$-coloring} chain $\cbra{\bm{X}^{\ucc}_t}_{t\geq 0}$ on the set of colorings $\Theta_{k, N}$ is given by the following distribution on $\bm{X}_{t+1}^{\ucc}|\bm{X}_{t}^{\ucc}$. To sample $\bm{X}_{t+1}^{\ucc} =(\bm{y}_1,\dots,\bm{y}_k)$ given the current state $\bm{X}_{t}^{\ucc}=x=(x_1,\dots,x_{k})$ uniformly sample $\bm{i}\in [k]$ and $\bm{\ell}\in [N]$ and set
    \begin{align*}
        \bm{y}_{j}=& \begin{cases}
            \bm{\ell} & j = \bm{i} \\
            x_{\bm{i}} & \bm{\ell} = x_j \\
            x_j & \text{otherwise} \\
        \end{cases}.
    \end{align*}
    Let $P^{\ucc}_{k, N}$ be the transition matrix for this Markov chain.
\end{definition}

We call this the uniform clique coloring chain, since at every step a random vertex $\bm{i}$ is re-colored with a uniformly random color from the entire set $[N]$. If this color is already taken by another vertex $j$, the two vertices swap colors. This additional symmetry allows us to obtain a bound on the log-Sobolev constant of this chain by adapting the martingale method of Lee and Yau~\cite{lee1998logarithmic}. Moreover, it is not hard to relate the log-Sobolev constants of the uniform and standard clique coloring chains using the comparison method.

With all of our Markov chains defined, we now state the sequence of inequalities that will allow us to conclude \Cref{thm:main}, deferring the proofs of the auxiliary results to later sections.

\begin{theorem}\label{thm:log-sobolev-rev-circuits}
    Let $P^{\mathsf{rev}}_{k, n}$ be the transition matrix corresponding to the random walk from \Cref{def:rev-circuit-chain}. Then
    \begin{align*}
        \alpha(P^{\mathsf{rev}}_{k, n})\geq \Omega\pbra{\frac{1}{n^3k}}.
    \end{align*}
\end{theorem}
\begin{proof}
We will show the following sequence of inequalities (recall that $\gtrsim$ hides constant factors):
\begin{align*}
    \alpha(P^{\mathsf{rev}}_{k, n})
    \underset{\text{\Cref{cor:circuits to cc}}}{\gtrsim}& {\frac{1}{n^2} }\cdot \alpha(P^\cc_{k, 2^n})
    \underset{\text{\Cref{lem:compare-clique-colorings}}} {\gtrsim}
    {\frac{1}{n^2} }\cdot \alpha(P^\ucc_{k, 2^n}) 
    \underset{\text{\Cref{lem:log-sobolev-uniform-clique-coloring}}}{\gtrsim}
    {\frac{1}{ n^3k}}.\qedhere
\end{align*}
\end{proof}

\Cref{thm:log-sobolev-rev-circuits} immediately gives a mixing time of $\wt{O}(n^3k \cdot \log(1/\epsilon))$ for the reversible circuits chain by \Cref{thm:log-sob-to-mixing}; in \Cref{sec:generic states} we improve the mixing time to $\wt{O}(nk\cdot \log(1/\epsilon))$ by applying ideas of \cite{BH08}, thus proving \Cref{thm:main}.

It may then seem that \Cref{thm:log-sobolev-rev-circuits} is strictly weaker than \Cref{thm:main}. However, the proof of \Cref{thm:main} does not yield a good log-Sobolev inequality for the reversible circuits Markov chain. Thus we cannot use that proof to conclude results about pointwise convergence as we can from log-Sobolev bounds using \Cref{thm:pointwiseconvergence}, such as the following result:

\begin{corollary}\label{cor:pointwiseconvergence}
    Let $p^t_x$ be the distribution over $V$ after $t \gtrsim n^3k\pbra{\log nk + \log \frac{1}{\varepsilon}}$ steps of $P^{\mathsf{rev}}_{k,n}$. For all $x,y, \in V$
    \[
        \frac{1-\epsilon}{2^n(2^n -1)\cdots(2^n - k+1)} \leq \Pr[p^t_x = y] \leq \frac{1 + \epsilon}{2^n(2^n -1)\cdots(2^n - k+1)}.
    \]
\end{corollary}

\section{The Log-Sobolev Constant of the Uniform Clique Coloring Chain}

The goal of this section is to lower bound the log-Sobolev constant of the uniform clique coloring Markov chain.

Recall that the uniform $k$-clique $N$-coloring Markov chain has state space $\Theta_{k, N}$ of size $N(N-1)\dots (N-k+1)$. Given some $x = (x_1, \dots, x_k) \in \Theta_{k, N}$, the action of choosing vertex $i\in[k]$ and coloring it with color $\ell\in[N]$ (where this color can already exist in the clique, as per~\Cref{def:uniform-clique-coloring}) will be denoted by $x^{i, \ell}$. Namely
$$x^{i, \ell} \coloneq
\begin{cases}
    (\dots, x_{i-1}, \ell, x_{i+1}, \dots) &\text{ if }\ell \not\in x\\
    (\dots, x_{j-1}, x_i, x_{j+1} \dots, x_{i-1}, x_j, x_{i+1}, \dots) &\text{ if } \ell = x_j.
\end{cases}$$

Let $f:\Theta_{k, N} \to \mathbb{R}$ be a function on the state space of this chain. Since the stationary distribution is the uniform, the expectation of $f$ over its state space is
\begin{align*}
    \Ex_{\Theta_{k, N}}[f] \coloneq \frac{1}{|\Theta_{k, N}|} \sum_{x \in \Theta_{k, N}} f(x).
\end{align*}
Moreover, the Dirichlet form of this chain can be written as
\begin{align*}
    \mathcal{E}_{P^{\ucc}_{k, N}}(\sqrt{f}, \sqrt{f}) &= \frac{1}{2}\Ex_{x\in \Theta_{k, N}}\sbra{ \Ex_{i \in [k]} \sbra{\Ex_{\ell \in [N]} \sbra{\pbra{\sqrt{f(x^{i, \ell})} - \sqrt{f(x)}}^2}}}\\
    &= \frac{1}{2kN\cdot |\Theta_{k, N}|}\sum_{x\in \Theta_{k, N}} \sum_{i \in [k]} \sum_{\ell \in [N]} \left(\sqrt{f(x^{i, \ell})} - \sqrt{f(x)}\right)^2.
\end{align*}

With this notation in mind, we now prove that this Markov chain has a large log-Sobolev constant. 
\begin{lemma}
    \label{lem:log-sobolev-uniform-clique-coloring}
    The log-Sobolev constant of the uniform $k$-clique $N$-coloring Markov chain satisfies
    $$\alpha(P^{\ucc}_{k, N}) \geq \frac{1}{12k \log N}$$
    when $k\leq N/2$.
\end{lemma}


\begin{proof}
    Our starting point is the recursive structure of the uniform clique coloring problem, which allows us to apply the martingale method of \cite{lee1998logarithmic}. In particular, let $x$ be uniformly distributed over the state space $\Theta_{k, N}$. Then if we condition on the $i^{th}$ vertex having color $\ell$, the distribution of the colors of the remaining $k-1$ vertices is isomorphic to the uniform distribution over $\Theta_{k-1, N-1}$, the state space of the uniform $(k-1)$-clique $(N-1)$-coloring Markov chain.
    
    For any vertex $i\in [k]$ and color $c\in[N]$ define the \emph{conditional function}
    \begin{align*}
        f_{i,c}:\cbra{(x_1,\dots,x_k)\in\Theta_{k,N}:x_i=c}\to \R
    \end{align*}
    to be simply the restriction of $f$ to this domain: $f_{i,c}(x)=f(x)$ for all $x\in \Theta_{k,N}$ with $x_i=c$. Since $\cbra{(x_1,\dots,x_k)\in\Theta_{k,N}:x_i=c}$ is isomorphic to $\Theta_{k-1,N-1}$, by a slight abuse of notation we also regard $f_{i,c}:\Theta_{k-1,N-1}\to \R$.
    
    Moreover, for every vertex $i\in[k]$, define the \emph{marginal function} $F_i:[N]\to \R$ by defining for every color $c\in[N]$
    \begin{align*}
        F_i(c) \coloneq  \Ex_{\substack{\bm{x} \in \Theta_{k, N} \\ \bm{x}_i = c}}[f(\bm{x})].
    \end{align*}
    The \emph{chain rule of conditional entropy} (\cite{Sal20}, Equation 13) implies that for any $i\in[k]$,
    \begin{equation}
        \label{eq:chain-rule}
        \ent(f) = \Ex_{\bm{c}}[\ent(f_{i,\bm{c}})] + \ent\pbra{F_i}.  
    \end{equation}
    By summing over all vertices $i\in[k]$, we get
    \begin{align}
        \label{eq:summation-two-terms}
        k\cdot \ent(f) = \sum_{i \in [k]} \Ex_{\bm{c}_i}[\ent(f_{i, \bm{c}_i)}] + \sum_{i \in [k]}\ent\left(F_i\right).  
    \end{align}
    We bound the two summations of the right-hand side separately in~\Cref{cl:first-sum} and~\Cref{cl:second-sum} and conclude that
    \begin{gather*}
        k\cdot \ent(f) \leq \frac{kN}{N-1} \cdot \alpha(P^{\ucc}_{k-1, N-1})^{-1} \cdot \mathcal{E}_{P^{\ucc}_{k, N}}(\sqrt{f}, \sqrt{f}) + 3k\log N \cdot \mathcal{E}_{P^{\ucc}_{k, N}}(\sqrt{f}, \sqrt{f}). \\
        \implies \ent(f) \leq \sbra{\frac{N}{N-1}\cdot\alpha(P^{\ucc}_{k-1, N-1})^{-1} + 3\log N} \cdot \mathcal{E}_{P^{\ucc}_{k, N}}(\sqrt{f}, \sqrt{f}).
    \end{gather*}
    This gives us a recurrence relation for the log-Sobolev constant of the uniform clique coloring chain. For every $k$ and $N$, we have
    \begin{align}
        \label{eq:recurrence-relation}
        \alpha(P^{\ucc}_{k, N})^{-1} &\leq \frac{N}{N-1}\cdot \alpha(P^{\ucc}_{k-1, N-1})^{-1} + 3\log N.
    \end{align}
    We proceed to solve this recurrence via induction. For fixed integers $k_{\max}$ and $N_{\max}$, we will prove that for all $1\leq k\leq k_{\max}$,
    $$\alpha(P^{\ucc}_{k, N_{\max}-k_{\max}+k})^{-1}\leq 6\cdot \frac{N_{\max}-k_{\max}+k}{N_{\max}-k_{\max}}\cdot k\log N_{\max}.$$
    
    For the base case of $k=1$, we observe that uniform $1$-clique $(N_{\max} - k_{\max} + 1)$-coloring has transition probabilities that correspond to the complete graph over $N_{\max} - k_{\max} + 1$ vertices. We use known results for the log-Sobolev constant of the complete graph (\cite{DSC96}, Corollary A.4) to deduce that
    \begin{align*}
        \alpha(P^{\ucc}_{1, N_{\max}-k_{\max}+1})^{-1}\leq 3\log(N_{\max}-k_{\max}+1)\leq 6 \log N_{\max}.
    \end{align*}
    Now let $k\geq 2$ and assume that the claim holds for all $k'\leq k$. Then using \Cref{eq:recurrence-relation} we find
    \begin{align*}
        \alpha(P^{\ucc}_{k, N_{\max}-k_{\max}+k})^{-1}
        &\leq \frac{N_{\max}-k_{\max}+k}{N_{\max}-k_{\max}+k-1}\cdot \alpha(P^{\ucc}_{k-1,N_{\max}-k_{\max}+k-1})^{-1}+3\log(N_{\max}-k_{\max}+k)\\
        &=6\cdot \frac{N_{\max}-k_{\max}+k}{N_{\max}-k_{\max}}\cdot (k-1)\log N_{\max}+3\log(N_{\max}-k_{\max}+k)\\
        &\leq 6\cdot \frac{N_{\max}-k_{\max}+k}{N_{\max}-k_{\max}}\cdot k\log N_{\max}.
    \end{align*}
    In the above calculation, we used the fact that $k_{\max}\leq N_{\max}/2$, and that $N_{\max}$ is at least some fixed constant. This finishes the inductive proof, and by setting $k=k_{\max}$ we obtain the desired bound.
\end{proof}

It remains to prove the two claims used in the proof of \Cref{lem:log-sobolev-uniform-clique-coloring}.
\begin{claim}\label{cl:first-sum}
    For any $f:\Theta_{k,N}\to \R$ we have
    \begin{align*}
        \sum_{i \in [k]} \Ex_{\bm{c}_i} \left[\ent(f_{i,\bm{c}_i}) \right] \leq \frac{kN}{N-1} \cdot \alpha(P^{\ucc}_{k-1, N-1})^{-1}\cdot \mathcal{E}_{P^{\ucc}_{k, N}}(\sqrt{f}, \sqrt{f}).
    \end{align*}
\end{claim}
\begin{proof}
    Recall that when we condition $f$ on vertex $i$ having color $c_i$, its domain is isomorphic to the state space of the uniform $(k-1)$-clique $(N-1)$-coloring chain. The log-Sobolev constant of this smaller restricted chain implies that
    \begin{align*}
        \ent(f_{i, c_i}) &\leq \alpha(P^{\ucc}_{k-1, N-1})^{-1}\cdot \mathcal{E}_{P^{\ucc}_{k-1, N-1}}\left(\sqrt{f_{i, c_i}}, \sqrt{f_{i, c_i}}\right).
    \end{align*}
    Our goal is to relate the Dirichlet form of $P^{\ucc}_{k-1, N-1}$ to the Dirichlet form of $P^{\ucc}_{k, N}$. We start by expanding the right-hand side while keeping in mind that $f_{i, c_i}$ has fixed the color of vertex $i$ to $c_i$.
    \begin{align*}
        \ent(f_{i, c_i}) &\leq \frac{\alpha(P^{\ucc}_{k-1, N-1})^{-1}}{2(N-1)(k-1)|\Theta_{k-1, N-1}|} \sum_{\substack{x \in \Theta_{k, N} \\ x_{i} = c_i}}
        \sum_{\substack{j \in [k] \\ j \neq i}}
        \sum_{\substack{\ell \in [N] \\ \ell \neq c_i}} \left(\sqrt{f(x^{j, \ell})} - \sqrt{f(x)}\right)^2  
    \end{align*}
    Let us take the expectation now over all $N$ values of $c_i$. We note that the log-Sobolev of $P^{\ucc}_{k-1, N-1}$ is not dependent on the value of $c_i$ due to symmetry, thus we factor it outside the summation.
    \begin{align*}
        \Ex_{\bm{c}_i} \left[\ent(f_{i, \bm{c}_i})\right] \leq \frac{\alpha(P^{\ucc}_{k-1, N-1})^{-1}}{2N(N-1)(k-1) |\Theta_{k-1, N-1}|} \sum_{c_i\in [N]} \sum_{\substack{x \in \Theta_{k, N} \\ x_{i} = c_i}}
        \sum_{\substack{j \in [k] \\ j \neq i}}
        \sum_{\substack{\ell \in [N] \\ \ell \neq c_i}} \left(\sqrt{f(x^{j, \ell})} - \sqrt{f(x)}\right)^2.
    \end{align*}
    Summing over all $i \in [k]$ yields the following
    \begin{align*}
        \sum_{i \in [k]} \Ex_{\bm{c}_i} \sbra{\ent(f_{i, \bm{c}_i})} &\leq \frac{\alpha(P^{\ucc}_{k-1, N-1})^{-1}}{2N(N-1)(k-1)|\Theta_{k-1, N-1}|} \sum_{i \in [k]} \sum_{c_i\in [N]}
        \sum_{\substack{x \in \Theta_{k, N} \\ x_{i} = c_i}}
        \sum_{\substack{j \in [k] \\ j \neq i}}
        \sum_{\substack{\ell \in [N] \\ \ell \neq c_i}} \pbra{\sqrt{f(x^{j, \ell})} - \sqrt{f(x)}}^2.
    \end{align*}
    Notice that each tuple $x$ is counted $k$ times in the summation of the right-hand side, one time for each $(i, c_i)$ that satisfies $c_i = x_i$. Then each $\pbra{\sqrt{f(x^{j^\prime, \ell})} - \sqrt{f(x)}}^2$ term appears at most $(k-1)$ times, since out of the $k$ times that $x$ appears, one of them satisfies $j^\prime = i$, and thus it does not contribute to the sum.
    
    This implies that the sum above is at most $(k-1)$ times the summation that corresponds to the Dirichlet form of $\mathcal{E}_{P^{\ucc}_{k, N}}$.
    \begin{align*}
        \sum_{i \in [k]} \Ex_{\bm{c}_i} \sbra{\ent(f_{i, \bm{c}_i})}
        &\leq \frac{\alpha(P^{\ucc}_{k-1, N-1})^{-1}}{2N(N-1)(k-1)|\Theta_{k-1, N-1}|} \cdot (k-1) \cdot  2kN\cdot |\Theta_{k, N}|\cdot \mathcal{E}_{P^{\ucc}_{k, N}}(\sqrt{f}, \sqrt{f}) \\
        &= \frac{kN}{N-1}\cdot \alpha(P^{\ucc}_{k-1, N-1})^{-1}\cdot \mathcal{E}_{P^{\ucc}_{k, N}}(\sqrt{f}, \sqrt{f}).  
    \end{align*}
\end{proof}

\begin{claim}
    \label{cl:second-sum}
    Let $f:\Theta_{k, N} \to \mathbb{R}$ be a function, and for all $i\in [k]$, $F_i:[N]\to \R$ is the $i^{th}$ marginal function of $f$ that maps color $c$ to
    $F_i(c) \coloneq  \Ex_{\bm{x} \in \Theta_{k, N}, \bm{x}_i = c}[f(\bm{x})]$.
    Then it holds that
    \begin{align*}
        \sum_{i=1}^k \ent(F_i) \leq k\log N \cdot \mathcal{E}_{P^{\ucc}_{k, N}}(\sqrt{f}, \sqrt{f}).
    \end{align*}
\end{claim}
\begin{proof}
    Consider the random walk on the set $[N]$ of colors where at every step we move to a uniformly random color (including the color we are currently in). The transition matrix of this walk is the complete graph over $N$ vertices and we denote it by $P^{\compl}_N$. Let us apply the log-Sobolev inequality of $P^{\compl}_N$ to the function $F_i$:
    \begin{align*}
        \ent\left(F_i\right) &\leq \alpha(P^{\compl}_N)^{-1} \cdot \mathcal{E}_{P^{\compl}_N}\left(\sqrt{F_i}, \sqrt{F_i}\right)\\
        &= \frac{\alpha(P^{\compl}_N)^{-1}}{2N^2}\cdot \sum_{\ell \in [N]} \sum_{\ell' \in [N]} \left(\sqrt{F_i(\ell')} - \sqrt{F_i(\ell)}\right)^2. \numberthis \label{eq:cl-second-sum-eq1}
    \end{align*}
    We would like to rewrite the Dirichlet form of $P^{\compl}_N$ in terms of $P^{\ucc}_{k, N}$. We start by expanding the definition of $F_i$
    \begin{align*}
        \pbra{\sqrt{F_i(\ell')} - \sqrt{F_i(\ell)}}^2
        = \pbra{\sqrt{\Ex_{\substack{\bm{x} \in \Theta_{k, N} \\\bm{x}_i = \ell'}} \sbra{f(\bm{x})}} - \sqrt{\Ex_{\substack{\bm{x} \in \Theta_{k, N}\\ \bm{x}_i = \ell}} \sbra{f(\bm{x})}}}^2.
    \end{align*}
    Observe that sampling a random $\bm{x} \in \Theta_{k, N}$ such that $\bm{x}_i = \ell'$, is equivalent to sampling a random $\bm{x}$ with $\bm{x}_i = \ell$, and then outputting $\bm{x}^{i, \ell'}$:
    \begin{align*}
        \pbra{\sqrt{F_i(\ell')} - \sqrt{F_i(\ell)}}^2 = \pbra{\sqrt{\Ex_{\substack{\bm{x} \in \Theta_{k, N} \\ \bm{x}_i = \ell}} \sbra{f(\bm{x}^{i, \ell'})}} - \sqrt{\Ex_{\substack{\bm{x} \in \Theta_{k, N} \\\bm{x}_i = \ell}} \sbra{f(\bm{x})}}}^2.
    \end{align*}
    Since the function on the right-hand side is convex, Jensen's inequality implies that
    \begin{align*}
        \pbra{\sqrt{F_i(\ell')} - \sqrt{F_i(\ell)}}^2 
        \leq \Ex_{\substack{\bm{x} \in \Theta_{k, N} \\ \bm{x}_i = \ell}} \sbra{\pbra{\sqrt{f(\bm{x}^{i, \ell'})} - \sqrt{f(\bm{x})}}^2}.
    \end{align*}
    Plugging in the above inequality to~\Cref{eq:cl-second-sum-eq1} we get
    \begin{align*}
        \ent(F_i) \leq \frac{\alpha(P^{\compl}_N)^{-1}}{2N^2}\sum_{\ell \in [N]} \sum_{\ell' \in [N]} \Ex_{\substack{\bm{x} \in \Theta_{k, N} \\ \bm{x}_i = \ell}} \sbra{\pbra{\sqrt{f(\bm{x}^{i, \ell'})} - \sqrt{f(\bm{x})}}^2}.
    \end{align*}
    We sum over all $i \in [k]$ to get
    \begin{align*}
        \sum_{i=1}^k \ent(F_i) &\leq \frac{\alpha(P^{\compl}_N)^{-1}}{2N^2}\sum_{\ell \in [N]} \sum_{\ell' \in [N]} \sum_{i \in [k]} \Ex_{\substack{\bm{x} \in \Theta_{k, N} \\ \bm{x}_i = \ell}} \sbra{\pbra{\sqrt{f(\bm{x}^{i, \ell'})} - \sqrt{f(\bm{x})}}^2 }\\
        &= \frac{\alpha(P^{\compl}_N)^{-1}}{2N|\Theta_{k, N}|} \sum_{\ell \in [N]} \sum_{\ell' \in [N]} \sum_{i \in [k]} \sum_{\substack{x \in \Theta_{k, N} \\ x_i = \ell}} \sbra{\pbra{\sqrt{f(x^{i, \ell'})} - \sqrt{f(x)}}^2}.
    \end{align*}
    The right-hand side now contains all $\pbra{\sqrt{f(x^{i, \ell'})} - \sqrt{f(x)}}^2$ terms that appear in $\mathcal{E}_{P^{\ucc}_{k, N}}$ exactly once. Thus we can substitute this Dirichlet form (and adjust its scaling). Moreover, the log-Sobolev constant of the complete graph over $N$ vertices is well-studied and satisfies $\alpha(P^{\compl}_N)^{-1} \leq 3\cdot \log N$ (\cite{DSC96}, Corollary A.4). We conclude that
    \begin{align*}
        \sum_{i=1}^k \ent(F_i) &\leq 3k\log N \cdot \mathcal{E}_{P^{\ucc}_{k, N}}(\sqrt{f}, \sqrt{f}).\qedhere
    \end{align*}
\end{proof}

\section{The Log-Sobolev Constant of the Standard Clique Coloring Chain}

The goal of this section is to translate the log-Sobolev bound from the uniform clique coloring chain~\Cref{lem:log-sobolev-uniform-clique-coloring} to the standard clique coloring chain. Since the two chains are very similar, applying the comparison method is a natural approach.

\begin{lemma}
    \label{lem:compare-clique-colorings}
    The log-Sobolev constant of the $k$-clique $N$-coloring Markov chain satisfies
    \begin{align*}
        \alpha(P^{\cc}_{k, N}) \geq& \frac{1}{19}\cdot \alpha(P^{\ucc}_{k, N}).
    \end{align*}
\end{lemma}
\begin{proof}
    Define the following (randomized) map $\bm{\Delta}$ that maps edges of $P^{\ucc}_{k, N}$ to paths in $P^{\cc}_{k, N}$. Each edge of $P^{\ucc}_{k, N}$ that connects $x$ and $x^{i, \ell}$ is determined by a vertex $x \in \Theta_{k, N}$ and the pair $(i, \ell) \in [k]\times [N]$. We assign to this edge a path in $P^{\cc}_{k, N}$ drawn according to the following distribution:
    \begin{align*}
        \bm{\Delta}(x, x^{i, \ell})=&
        \begin{cases}
            (x, x^{i, \ell})& \ell \not \in x \setminus \{x_i\} ,\\
            (x, \underbrace{x^{i, \bm{\ell'}}}_{y}) \mid \mid (y, \underbrace{y^{j, x_i}}_{z}) \mid \mid (z, z^{i, x_j}) & \ell = x_j~\text{for}~j \neq i,~~\bm{\ell'} \sim [N]\setminus x.
        \end{cases}
    \end{align*}
    Here the symbol ``$\mid\mid$" denotes the concatenation of edges to make a path.
    Intuitively, the path assigned to edge $(x, x^{i, \ell})$ is either itself (whenever $(x, x^{i, \ell}$ is also an edge of $P^{\cc}_{k, N}$), or a sequence of three edges that swap the colors $x_i$ and $x_j$ by using a random unused color $\bm{\ell'}$.

    Now we bound the comparison constant $A(\bm\Delta)$.
    \begin{align*}
        A(\bm{\Delta}) = \max_{\substack{(a,b)\in E^{\cc}}} \cbra{\frac{1}{\pi^{\cc}(x)P^{\cc}(a,b)}\sum_{(x,y)\in E^{\ucc}}
        \Ex_{\bm{\Delta}}\sbra{\mathbf{1}_{(a, b) \in \bm{\Delta}(x,y)}\cdot |\bm{\Delta}(x,y)|}\cdot  \pi^{\ucc}(x)\cdot P^{\ucc}(x, y)}
    \end{align*}
    The stationary distributions of both chains are the uniform over $\Theta_{k, N}$, and thus the stationary probabilities cancel.
    \begin{align*}
        A(\bm{\Delta}) &= \max_{\substack{(a,b)\in E^{\cc}}} \cbra{k(N-k+1)\sum_{(x,y)\in E^{\ucc}}
        \Ex_{\bm{\Delta}}\sbra{\mathbf{1}_{(a, b) \in \bm{\Delta}(x,y)}\cdot |\bm{\Delta}(x,y)|}\cdot  \frac{1}{kN}} \\
        &= \max_{\substack{(a,b)\in E^{\cc}}} \cbra{\frac{N-k+1}{N}\sum_{(x,y)\in E^{\ucc}}
        \Ex_{\bm{\Delta}}\sbra{\mathbf{1}_{(a, b) \in \bm{\Delta}(x,y)}\cdot |\bm{\Delta}(x,y)|}}.
    \end{align*}

    Our goal will be to bound the sum of expectations. First, let us partition the paths into the ones with length $1$ and length $3$. To do that, we observe that the length of each path $\bm{\Delta}(x, y)$ is deterministic and only depends on $x$ and $y$.
    \begin{align*}
        \sum_{(x,y)\in E^{\ucc}}\Ex_{\bm{\Delta}} \sbra{\mathbf{1}_{(a, b) \in \bm{\Delta}(x,y)}\cdot |\bm{\Delta}(x,y)|} &=
        \sum_{\substack{(x,y)\in E^{\ucc} \\ |\bm{\Delta}(x, y)| = 1}} \Ex_{\bm{\Delta}}\sbra{\mathbf{1}_{(a, b) \in \bm{\Delta}(x,y)}}
        + 3\sum_{\substack{(x,y)\in E^{\ucc} \\ |\bm{\Delta}(x, y)| = 3}} \Ex_{\bm{\Delta}}\sbra{\mathbf{1}_{(a, b) \in \bm{\Delta}(x,y)}}.
    \end{align*}
    We can now easily bound the first term. For a path with a single edge to include $(a, b)$, it must hold that $(x, y) = (a, b)$. Thus the first term is at most $1$. To bound the second term, we consider the location $t \in \{1, 2, 3\}$ where $(a, b)$ appears in $\bm{\Delta}(x, y)$. We write $(a, b) = \bm{\Delta}(x, y)_t$ if $(a, b)$ appears as the $t^{th}$ edge of the path. Formally,
    \begin{align*}
        \sum_{(x,y)\in E^{\ucc}}\Ex_{\bm{\Delta}} \sbra{\mathbf{1}_{(a, b) \in \bm{\Delta}(x,y)}\cdot |\bm{\Delta}(x,y)|} &\leq
        1 + 3\sum_{t \in \{1, 2, 3\}} \sum_{\substack{(x,y)\in E^{\ucc} \\ |\bm{\Delta}(x, y)| = 3}} \Ex_{\bm{\Delta}}\sbra{\mathbf{1}_{(a, b) = \bm{\Delta}(x,y)_t}}.
    \end{align*}
    Observe now that once we fix the $t^{th}$ edge to be $(a, b)$, there are only $k-1$ possible $3$-edge paths. This is because our map $\bm{\Delta}$ performs three transpositions between the elements $x_i, x_j, \bm{\ell}'$. The edge $(a, b)$ specifies two of the elements, and the third element is one of the remaining $k-1$ elements of the tuples at the endpoints of $(a, b)$. Once this third element is specified, the edge $(x, y)$ and its respective path $\bm{\Delta}(x, y)$ is fully determined.

    Each $3$-edge path has a probability of $\frac{1}{N-k}$ to appear, since it depends on the random choice of $\bm{\ell'}$ from the set $[N] \setminus x$. Thus we bound the expectation above to be at most
    \begin{align*}
        \sum_{(x,y)\in E^{\ucc}}\Ex_{\bm{\Delta}} \sbra{\mathbf{1}_{(a, b) \in \bm{\Delta}(x,y)}\cdot |\bm{\Delta}(x,y)|} &\leq
        1 + \frac{9(k-1)}{N-k}.
    \end{align*}
    We conclude that the comparison constant of $\bm{\Delta}$ is
    \begin{align*}
        A(\bm{\Delta})
        &= \max_{\substack{(a,b)\in E^{\cc}}} \cbra{\frac{N-k+1}{N}\sum_{(x,y)\in E^{\ucc}}
        \Ex_{\bm{\Delta}}\sbra{\mathbf{1}_{(a, b) \in \bm{\Delta}(x,y)}\cdot |\bm{\Delta}(x,y)|}} \\
        &\leq \frac{N-k+1}{N} \left(1 + \frac{9(k-1)}{N-k}\right) \\
        &= \frac{N-k+1}{N} + \frac{9(k-1)}{N}\cdot \frac{N-k+1}{N-k} \\
        &\leq 1 + 9\cdot 2 = 19.\qedhere
    \end{align*}
\end{proof}

Our log-Sobolev bound for the standard clique-coloring chain now follows directly from \Cref{lem:log-sobolev-uniform-clique-coloring} and \Cref{lem:compare-clique-colorings}.
\begin{corollary}\label{cor:log-sobolev-clique-coloring}
    The log-Sobolev constant of the $k$-clique $N$-coloring Markov chain satisfies
    \begin{align*}
        \alpha(P^{\cc}_{k, N}) \geq \Omega\pbra{\frac{1}{k\log N}}.
    \end{align*}
\end{corollary}

\subsection{Clique-Coloring Walk to Random Circuits Walk}


We would like to transfer our log-Sobolev constant bound of the $k$-clique $N$-coloring Markov chain from~\Cref{cor:log-sobolev-clique-coloring}, to the random circuits Markov chain.
This is done via the randomized paths construction of Brodsky and Hoory to compare this walk to clique coloring.

\begin{lemma}[\cite{BH08}]\label{lem:circuits to cc}
    When $k\leq 2^n/3$ there exists a randomized map $\bm{\Phi}$ that takes as input an edge $(x, y)$ of $P^{\cc}_{k, 2^n}$ and outputs a sequence of edges in $P^{\mathsf{rev}}_{k, n}$ connecting $x$ and $y$ such that the comparison constant satisfies
    \begin{align*}
        A(\bm{\Phi}) = O(n^2).
    \end{align*}
\end{lemma}

\begin{corollary}\label{cor:circuits to cc}
    If $k\leq 2^n/3$ then
    \begin{align*}
        \alpha(P_{\mathsf{rev}})\gtrsim {\frac{1}{n^2} }\cdot \alpha(P_\cc).
    \end{align*}
\end{corollary}
\begin{proof}
    This follows immediately from \Cref{lem:circuits to cc} and \Cref{thm:comparison}. 
\end{proof}

\section{Even Faster Mixing of the Random Circuits Walk via Generic States}\label{sec:generic states}
We can improve the dependence on $n$ of the mixing time of the random reversible circuits Markov chain $P^{\textsf{rev}}_{k, n}$ from cubic to linear using an idea of~\cite{BH08}. The main observation is that after $n \cdot \text{polylog}\left(n, k\right)$ steps of $P^{\mathsf{rev}}_{k, n}$, the chain is very likely to be in a \emph{generic} state, that is a state where no two of the bit-strings agree on many bits. Generic states happen with good probability and are nicer to work with, thus when we restrict our Markov chain $P^{\mathsf{rev}}_{k, n}$ to generic states we apply the comparison theorem with a better (logarithmic) comparison constant.

\begin{definition}[Generic states,~\cite{BH08}]
    Let $w = \left\lceil 10 \cdot \left( \log k + \log n \right) \right\rceil, p = \left\lceil \frac{n}{2w} \right\rceil$. Let $C_1,\cdots C_p, C$ be a partition of $[n]$ such that $|C_t| = w$ for $t \in [p]$, and $|C| = n - pw$. A state $\left(x_1, \cdots, x_k\right)$ is \textit{generic} if for $i \neq i'$, $x_i$ and $x_{i'}$ are distinct when restricted to a part $C_t$ (but not $C$). Let $\mathsf{Generic}_{k,n}$ denote the set of generic states.
\end{definition}

In other words, we divide the $n$ bits of the input into two subsets $\bigcup_{t \in [p]} C_t$ and $C$ of roughly equal size. Then we further divide the first subset into $p$ equal-length blocks that hold a logarithmic number of bits. A state is generic if no two distinct elements $x_i, x_{i'}$ are equal in any of the $C_t$ parts. Since we now deal with $n$-bit strings, we will extend our notation and write $x_{i, j}$ to denote the $j^{th}$ bit of the $i^{th}$ element of the state $x$.

We define below the \emph{generic state reversible circuit} Markov chain $P^{\mathsf{grev}}_{k, n}$ to be the restriction of $P^{\mathsf{rev}}_{k, n}$ to generic states. 

\begin{definition}[Generic state reversible circuit Markov chain]
    The matrix $P^{\mathsf{grev}}$ is the transition matrix of the Markov chain on $\mathsf{Generic}_{k,n}$ such that for any $x, y\in\mathsf{Generic}_{k,n}$,
    \begin{align*}
        P^{\mathsf{grev}}(x,y)=&\frac{P^{\mathsf{rev}}(x, y)}{\sum_{z\in \mathsf{Generic}_{k,n}}P^{\mathsf{rev}}(x,z)}.
    \end{align*}
\end{definition}

\begin{lemma}[\cite{BH08}, Equation (3)]\label{lem:Pgrev to Prev}
    There exists a constant $\epsilon>0$ such that if $\tau_\epsilon\pbra{P^{\mathsf{grev}}}\leq O(n^3k^3)$, and $k \leq 2^{n/50}$, then
    \begin{align*}
        \tau\pbra{P^{\mathsf{rev}}} \leq \tau_{\varepsilon}\pbra{P^{\mathsf{grev}}} + O(n \cdot\mathrm{polylog}\pbra{n,k}).
    \end{align*}
\end{lemma}

We bound the mixing time of the $P^{\mathsf{grev}}$ Markov chain by bounding its log-Sobolev constant. We use the comparison of~\cite{BH08} as stated in~\Cref{lem:Ptildegrev to Pgrev} to relate its log-Sobolev constant to the log-Sobolev constant of a related product chain on generic states, $\widetilde{P}^{\mathsf{grev}}$. We get our final estimate by bounding the log-Sobolev constant of the $\widetilde{P}^{\mathsf{grev}}$ Markov chain in~\Cref{lem:log sobolev Ptildegrev} using results for product chains from~\cite{DSC96}.

Below we introduce the $\widetilde{P}^{\mathsf{grev}}$ Markov chain.

\begin{definition}[Product chain on generic states]
    Let $\widetilde{P}^{\mathsf{grev}}$ be the Markov chain on state space $\mathsf{Generic}_{k,n}$, where to sample the next state $\bm{y}=(\bm{y}_1,\dots,\bm{y}_k)$ given the current state $x=(x_1,\dots,x_k)\in \mathsf{Generic}_{k,n}$ we do the following:
    \begin{itemize}
        \item With probability $\frac12$, toss a fair coin.
        \begin{itemize}
            \item If the coin has landed heads, set $\bm{y}=x$.
            \item Else, sample uniformly at random $\bm{c} \sim C, \bm{r} \sim [k]$ and set for all $i\in[k]$ and $j\in[n]$
            \begin{align*}
                \bm{y}_{i,j}=&
                \begin{cases}
                    x_{i,j} &\text{ if $i\neq \bm{r}$ or $j\neq \bm{c}$}\\
                    1-x_{i,j} &\text{ if $i=\bm{r}$ and $j=\bm{c}$.}
                \end{cases}
            \end{align*}
        \end{itemize}
        \item With probability $\frac12$, sample uniformly at random $\bm{\ell} \sim [p], \bm{r} \sim [k]$ and a random string $\bm{u}\in \{0,1\}^w$ such that $\bm{u}\neq x_{i,C_{\bm{\ell}}}$ for any $i\neq \bm{r}$. Set
        \begin{align*}
            \bm{y}_{i,j}=&
            \begin{cases}
                x_{i,C_\ell} &\text{ if $i\neq \bm{r}$ or $\ell\neq \bm{\ell}$}\\
                \bm{u} &\text{ if $i=\bm{r}$ and $\ell=\bm{\ell}$.}
            \end{cases}
        \end{align*}
    \end{itemize}
\end{definition}

Informally, given the current state $x$, one step of this Markov chain performs a change in exactly one of the two subsets of bits ($C$ or $\bigcup_{i \in [p]} C_i$) with equal probability. In the first case, it either flips the $\bm{c}^{th}$ bit from the subset $C$ of a random element $\bm{r}$ with probability $\frac{1}{2}$, or it does nothing. In the second case, it samples a uniformly random subset of bits $C_{\bm{\ell}}$ and replaces that subset with a new bit string $\bm{u}$ for a random element $\bm{r}$. All of the operations above are performed such that the resulting state remains generic.

It is not hard to observe that $\widetilde{P}^{\mathsf{grev}}$ is a \emph{product chain}, that is it acts ``independently'' on different parts of its state space. This means that we can compute its log-Sobolev constant by breaking it down into smaller chains.

\begin{definition}[Product Markov chain]
    \label{def:product-chain}
    Consider $t$ Markov chains $\{P_i\}_{i \in [t]}$ with state spaces $\{V_i\}_{i \in [t]}$ respectively. We define the \emph{product Markov chain} $\prod\pbra{\cbra{P_i}_{i \in [t]}}$ over the state space $\prod_{i \in [t]} V_i$ to be the Markov chain with transition matrix
    \begin{align*}
        \frac1t\sum_{i\in[t]}I \otimes \cdots \otimes P_i\otimes \cdots \otimes I.
    \end{align*}
    We will refer to the $P_i$'s as the \emph{factors} of $\prod\pbra{\cbra{P_i}_{i \in [t]}}$.
\end{definition}

\begin{lemma}[Log-Sobolev constant of product chain, Lemma 3.2 of~\cite{DSC96}]
    \label{lem:logsobolevproduct}
    The log-Sobolev constant of the product chain $\prod\pbra{\cbra{P_i}_{i \in [t]}}$ is related to the log-Sobolev constant of its factors as follows:
    \begin{align*}
        \alpha\pbra{\prod\pbra{\cbra{P_i}_{i \in [t]}}} = \frac{1}{t}\min_{i \in [t]} \alpha(P_i).
    \end{align*}
\end{lemma}

Using \Cref{lem:logsobolevproduct} we obtain the following bound by decomposing $\wt{P}^{\mathsf{grev}}$ into factor chains whose log-Sobolev constants are known.
\begin{lemma}\label{lem:log sobolev Ptildegrev}
    The following bound on the log-Sobolev constant of $\wt{P}^{\mathsf{grev}}$ holds:
    \begin{align*}
        \alpha\pbra{\widetilde{P}^{\mathsf{grev}}} \geq \Omega\pbra{\frac{1}{nk}}.
    \end{align*}
\end{lemma}
\begin{proof}
    We first write the state space $\mathsf{Generic}_{k,n}$ in the form of a product
    \begin{align*}
        \mathsf{Generic}_{k,n}=&~ \pbra{\prod_{i\in[p]}\Theta_{k,\{0,1\}^w}}\times \pbra{\{0,1\}^{k(n-wp)}}.\footnotemark
    \end{align*}
    Then decompose $\widetilde{P}^{\mathsf{grev}}$ as the product of two Markov chains $\prod\pbra{\cbra{\wt{P}_1, \wt{P}_2}}$. The first chain $\wt{P}_1$ corresponds to performing a change in the $\bigcup_{i \in [p]} C_i$ subset of the bits, and the second chain $\wt{P}_2$ corresponds to operating in the $C$ subset of the bits.

    \footnotetext{Recall from~\Cref{def:distinct-tuples} that $\Theta_{k,\{0,1\}^w}$ denotes the set of $k$-tuples of distinct elements of $\{0,1\}^w$.}
    
    \paragraph{The chain $\widetilde{P}_1$.} The state space of this chain is $\prod_{i\in[p]}\Theta_{k,\{0,1\}^w}$. We further decompose\footnote{We don't directly decompose $\widetilde{P}^{\mathsf{grev}}$ into all of its $t+1$ factors because to use \Cref{lem:logsobolevproduct} we need each factor of the product chain to have equal weight.}
    this chain as $\wt{P}_1=\prod\pbra{\cbra{\wt{P}_{1,\ell}}_{\ell \in [p]}}$, where $\tilde{P}_{1, \ell}$ corresponds to performing an operation on the $C_{\ell}$ subset of the bits.
    Thus the chain $\wt{P}_{1,\ell}$ has state space $\Theta_{k,\{0,1\}^w}$, since it corresponds to the size-$w$ subset $C_{\ell}$. To sample the next state $\bm{y}=(\bm{y}_1,\dots,\bm{y}_k)$ from the current state $\bm{x}=(x_1,\dots,x_k)$, we choose a random $\bm{i}\in[k]$ and a random $\bm{z}\in \{z \in \{0,1\}^w \mid z \notin x\} \cup \{x_{\bm{i}}\}$ and set for each $j\in[k]$
    \begin{align*}
        \bm{y}_j=
        \begin{cases}
            x_j & \text{ if $j\neq \bm{i}$.}\\
            \bm{z} & \text{if $j=\bm{i}$.}
        \end{cases}
    \end{align*}
    Notice that the transition matrix of this chain is equal to the transition matrix $P^{\cc}_{k, \{0, 1\}^w}$ of the standard $k$-clique $2^w$-coloring chain. Therefore, by \Cref{cor:log-sobolev-clique-coloring}, we have for all $\ell\in[p]$ that
    \begin{align*}
        \alpha\pbra{\wt{P}_{1,\ell}} \gtrsim \frac1{k\log{|\{0,1\}^w|}}   = \frac{1}{kw}.
    \end{align*}
    Applying \Cref{lem:logsobolevproduct}, we have 
    \begin{align}\label{eq:Ptilde1 log sobolev}
        \alpha\pbra{\wt{P}_1}=  \frac1p\min_{\ell\in[p]}\alpha\pbra{\wt{P}_{1,\ell}}\gtrsim \frac1p\cdot \frac1{kw}\gtrsim \frac1{nk}.
    \end{align}

    \paragraph{The chain $\widetilde{P}_2$.} We will ``flatten'' the bits from the subset $C$ of the $k$ elements into a sequence of $k(n-wp)$ bits. Then the $\wt{P}_2$ Markov chain corresponds to the random walk on the hypercube $\{0,1\}^{k(n - wp)}$ where to sample the next state $\bm{y}$ from the current state $x$ we sample $\bm{i}\in[k(n-wp)]$ uniformly at random and flip the $\bm{i}^{th}$ bit with probability $\frac12$. This chain is the product chain of $k(n - wp)$ chains on the space $\{0,1\}$ with transition probabilities $\frac{1}{2}$ to each state. We can write the transition matrix of $\wt{P}_2$ as the product
    \begin{align*}
        \prod\pbra{\cbra{\wt{P}_{2, \ell}}_{\ell \in [k(n-wp)]}},
    \end{align*}
    where each $\wt{P}_{2,\ell}$ is the $2 \times 2$ matrix with $\frac{1}{2}$'s. Equivalently, it corresponds to the transition matrix of the complete graph on two states. It is easy to see (e.g.~\cite{DSC96}, Corollary A.4) that $\alpha(\wt{P}_{2,\ell})\geq \frac1{3}$ for all $\ell$. Therefore, by 
    \Cref{lem:logsobolevproduct} we have
    \begin{align}\label{eq:Ptilde2 log sobolev}
        \alpha\pbra{\wt{P}_2}= \frac1{k(n-wp)}\min_{\ell\in[k(n-wp)]}\alpha\pbra{\wt{P}_{2,\ell}} \gtrsim \frac1{k(n-wp)}\gtrsim \frac1{nk}.
    \end{align}
    Applying \Cref{lem:logsobolevproduct} with \Cref{eq:Ptilde1 log sobolev} and \Cref{eq:Ptilde2 log sobolev} yields
    \begin{align*}
        \alpha\pbra{\widetilde{P}^{\mathsf{grev}}}  =&\frac12\min\cbra{\alpha\pbra{\widetilde{P}_1}, \alpha\pbra{\widetilde{P}_2}} = \Omega\pbra{\frac{1}{nk}}.\qedhere
    \end{align*}
\end{proof}

Armed with the log-Sobolev constant of $\wt{P}^{\mathsf{grev}}$, we employ the comparison method of~\cite{BH08} to bound the log-Sobolev constant of $\wt{P}^{\mathsf{rev}}$.

\begin{lemma}[\cite{BH08}, Lemma 16]\label{lem:Ptildegrev to Pgrev}
    There exists a randomized map $\bm{\Psi}$ that takes as input an edge $(x, y)$ of $\widetilde{P}^{\mathsf{grev}}$ and outputs a sequence of edges in $P^{\mathsf{grev}}$ connecting $x$ and $y$ with congestion $A(\bm{\Psi}) = \mathrm{polylog}(n,k)$. Consequently,
    \begin{align*}
        \alpha\pbra{{P}^{\mathsf{grev}}}\geq \frac{\alpha\pbra{\widetilde{P}^{\mathsf{grev}}}}{\polylog(n, k)}.
    \end{align*}
\end{lemma}

\begin{corollary}\label{cor:grev log sobolev}
    It holds that
    \begin{align*}
        \alpha(P_{\mathsf{grev}})\gtrsim {\frac{1}{nk\cdot \polylog(n, k)} }
    \end{align*}
\end{corollary}

Using now the well-known relation between the log-Sobolev constant and the mixing time of a Markov chain in total variation distance, we conclude: 

\main*

\begin{proof}
    Combining \Cref{lem:log sobolev Ptildegrev} and \Cref{lem:Ptildegrev to Pgrev} we find that $\alpha(P^{\mathsf{grev}})\geq \Omega\pbra{\frac1{nk}}$. This implies that for the constant $\epsilon'>0$ referenced in \Cref{lem:Pgrev to Prev}, we have $\tau_{\epsilon'}(P^{\mathsf{grev}})\leq O(nk\cdot \mathrm{polylog}(n,k))$. Then applying \Cref{lem:Pgrev to Prev} we have
    \begin{align*}
        \tau(P^{\mathsf{rev}})\leq&\tau_{\epsilon'}(P^{\mathsf{grev}})+O(nk\cdot\polylog(n,k))\leq O(nk\cdot\mathrm{polylog}(n,k)).
    \end{align*}
    Finally, we can decrease the total variation distance down to an arbitrary $\epsilon > 0$ by increasing the length of the walk by a multiplicative factor of $O(\log(1/\epsilon))$, and the statement follows.
\end{proof}

\section*{Acknowledgments}

We thank Thiago Bergamaschi, Tianren Liu, Stefano Tessaro, Vinod Vaikuntanathan, Alistair Sinclair, and Ryan O'Donnell for very helpful and insightful discussions.

\bibliographystyle{alpha}
\bibliography{references}

\end{document}